\documentclass[11pt, reqno]{amsart}
\usepackage{amsmath, amsthm, amscd, amsfonts, amssymb, graphicx, color}
\usepackage{mathrsfs}
\usepackage{enumitem}
\usepackage[bookmarksnumbered, colorlinks, plainpages]{hyperref}
\hypersetup{colorlinks=true, linkcolor=red, anchorcolor=red,
citecolor=red, urlcolor=red, filecolor=magenta, pdftoolbar=true}

\newtheorem{theorem}{Theorem}[section]
\newtheorem{lemma}[theorem]{Lemma}

\newtheorem{corollary}[theorem]{Corollary}
\theoremstyle{definition}

\newtheorem{example}[theorem]{Example}

\theoremstyle{remark}
\newtheorem{remark}[theorem]{Remark}
\numberwithin{equation}{section}

\newcommand{\OSR}{ \mathrm{Sch}}
\newcommand{\ve}{ \mathrm{vec}}
\newcommand{\s}{ \mathrm{span}}
\newcommand{\DS}{ \mathrm{dist}}
\newcommand{\SR}{ \mathrm{Sch}}
\newcommand{\RAN}{ \mathrm{ran}}
\newcommand{\rank}{ \mathrm{rank}}
\begin{document}

\title[Operator Schmidt Rank and the Geometry of Schmidt Profiles]{Operator Schmidt Rank and the Geometry of Schmidt Profiles}

\author[M. Kian, M.S. Moslehian ]{Mohsen Kian \MakeLowercase{and} Mohammad Sal Moslehian$^*$}

\address{Mohsen Kian: Department of Mathematics, University of Bojnord, P.O. Box
	1339, Bojnord 94531, Iran.}
\email{kian@ub.ac.ir; kian$_-$tak@yahoo.com}

\address{Mohammad Sal Moslehian: Department of Pure Mathematics, Faculty of Mathematical Sciences, Ferdowsi University of Mashhad, P. O. Box 1159, Mashhad 91775, Iran}
\email{moslehian@um.ac.ir; msmoslehian@gmail.com}



\thanks{$^*$Corresponding author}

\begin{abstract}
	The main result of the paper states that if
$0\neq v\in\mathbb C^m\otimes\mathbb C^n$, $d=\min\{m,n\}$, and
$p=\SR(v)$, then for every $1\le q\le d$,
\begin{equation*}
	\min\left\{\OSR(\mathbf U):
	\begin{array}{l}
		\mathbf U\in\mathbb M_m\otimes\mathbb M_n\text{ is unitary},\\
		\SR(\mathbf Uv)=q
	\end{array}\right\}
	=
	\left\lceil\max\left\{\frac pq,\frac qp\right\}\right\rceil .
\end{equation*}
Here $\OSR(\mathbf U)$ denotes the operator Schmidt rank of $\mathbf U$.
Furthermore, if $\mathbf U\in\mathbb M_m\otimes\mathbb M_n$ is unitary and
$\OSR(\mathbf U)\le k$, then, for every unit vector
$x\in\mathbb C^m\otimes\mathbb C^n$ and every $1\le r\le d$,
\begin{equation*}
	F_{\min\{kr,d\}}(\mathbf Ux)\ge F_r(x)
	\quad\text{and}\quad
	F_{\min\{kr,d\}}(x)\ge F_r(\mathbf Ux),
\end{equation*}
where $F_r(x)$ denotes the sum of the $r$ largest squared Schmidt
coefficients of $x$. This result is stronger than rank reachability as it
controls the entire cumulative Schmidt profile. We also establish a
majorization relation and use it to obtain sharp bounds on Schatten-norm
distortion and on the change of R\'enyi entropies under unitaries with
bounded operator Schmidt rank.
\end{abstract}

\maketitle
%----------------------------------------------------------%

\section{Introduction}
Tensor-product spaces enjoy several natural notions of rank for states and operators.
These ranks measure the complexity of states and operators relative to the underlying tensor-product structure. In quantum information, they are closely related to entanglement and to the nonlocal structure of quantum operations.

In the bipartite setting, pure states are characterized by their Schmidt coefficients, while operators admit an operator Schmidt decomposition.

For bipartite pure states, Schmidt coefficients play a central role in characterizing state transformations under restricted classes of operations.
Nielsen \cite{Nielsen1999} showed that a deterministic transformation $|\psi\rangle\to|\phi\rangle$ by local operations and classical communication is possible if and only if $\lambda(\psi)\prec\lambda(\phi)$, where $\lambda(\psi)$ and $\lambda(\phi)$ are the vectors of squared Schmidt coefficients.
Vidal \cite{Vidal2000} extended this result to probabilistic transformations and expressed the optimal conversion probability in terms of monotones defined by tail sums of the squared Schmidt coefficients.

There have been several works on the operator Schmidt rank of bipartite unitaries, mainly focused on the possible values of the rank and on the structure of unitaries with small operator Schmidt rank. Tyson \cite{Tyson2003} constructs bipartite unitaries with prescribed operator Schmidt ranks using Fourier methods. In particular, every rank from $1$ to $9$ occurs on $\mathbb C^3\otimes\mathbb C^3$. Cohen and Yu \cite{CohenYu2013} proved that every bipartite unitary of operator Schmidt rank two is locally equivalent to a controlled unitary. Chen and Yu \cite{ChenYu2014} studied the structure of Schmidt-rank-three unitaries, including controlled and block-controlled forms. M\"{u}ller-Hermes and Nechita \cite{MHN2018} later studied the general existence problem. They proved that, apart from the exceptional two-qubit case, all values $1,\ldots,\min\{n,m\}^2$
occur as operator Schmidt ranks of bipartite unitaries. Recent works have sought finer ways of measuring Schmidt numbers and $k$-positivity. For example, in \cite{Kian2026}, it is introduced a continuous refinement of the usual discrete Schmidt-number and $k$-positivity hierarchies.

Stahlke and Griffiths \cite{StGr2011} considered the deterministic implementation of bipartite unitaries by separable operations, including LOCC, using a bipartite pure entangled resource state. They showed that the Schmidt rank of the resource cannot be smaller than the operator Schmidt rank of the unitary, and that equality can occur only when the resource is maximally entangled on its Schmidt support.
 Cariello \cite{CAR} investigated how Schmidt rank constraints interact with other quantum properties, focusing on mutually unbiased bases and positive-partial-transpose entanglement. Abellanet-Vidal et al. \cite{ABE} provided a framework that bounds the entanglement of a bipartite quantum state from its spectrum alone.

Nielsen et al. \cite{Nielsen2003} investigated the nonlocal content of bipartite unitary dynamics, including operator Schmidt decompositions and Schmidt strength. Our entropy results are closely related in spirit to the idea that operator Schmidt rank limits entanglement generation. In addition, Bennett et al. \cite{BenHarLeSm2003} considered interactions as bidirectional channels and studied  the capacities for interaction Hamiltonians and nonlocal unitary gates to generate entanglement.  Chen and Yu \cite{ChenYu20162} studied the entangling power   of bipartite unitaries and show that   any bipartite permutation unitary of Schmidt rank three can be implemented using local operations and classical communications (LOCC) and two ebits.

\medskip

Throughout this paper, we denote the space of $m\times n$ complex matrices by $\mathbb M_{m,n}$ and use $\mathbb M_m$ for $\mathbb M_{m,m}$. We identify
$\mathbb M_m\otimes\mathbb M_n$ with the algebra of linear operators on
$\mathbb C^m\otimes\mathbb C^n$ under the canonical tensor-product representation $(A\otimes B)(u\otimes v):=A(u)\otimes B(v)\,\, A\in \mathbb M_m, B\in \mathbb M_n, u \in \mathbb C^m, v\in \mathbb C^n$. For a matrix $X\in\mathbb M_{m,n}$, we denote its singular values arranged in decreasing order and counted according to multiplicity as
$$
s_1(X)\geq s_2(X)\geq\cdots\geq s_d(X)\geq0,
$$
where, and throughout the paper, $d=\min\{m,n\}$. For $1\leq p<\infty$, the Schatten $p$-norm of $X$ is defined by
$\|X\|_p=\left(\sum_{j=1}^{d}s_j(X)^p\right)^{1/p}$, and the operator norm is given as $\|X\|_\infty=s_1(X)$. In particular, $\|X\|_2$ stands for the Hilbert--Schmidt norm.

Fix orthonormal bases $\{e_i\}_{i=1}^{m}$ of $\mathbb C^m$ and
$\{f_j\}_{j=1}^{n}$ of $\mathbb C^n$. We use the vectorization map
$$
\ve:\mathbb M_{m,n}\longrightarrow \mathbb C^m\otimes\mathbb C^n
$$
defined by
\begin{align*}
\ve(X)=\sum_{i=1}^{m}\sum_{j=1}^{n}
x_{ij}\,e_i\otimes f_j,\qquad X=(x_{ij})\in\mathbb M_{m,n}.
\end{align*}
In fact, the matrix unit $E_{ij}$, with a $1$ in position $(i,j)$ and $0$ elsewhere, corresponds with $e_i\otimes f_j$. This map is a linear isometry from $\mathbb M_{m,n}$, equipped with the
Hilbert--Schmidt norm, onto
$\mathbb C^m\otimes\mathbb C^n$, so that  $\|\ve(X)\|=\|X\|_2$.

For a vector $x\in \mathbb C^m\otimes\mathbb C^n$, a Schmidt decomposition of $x$ is a representation
$$
x=\sum_{i=1}^{r}s_i\,u_i\otimes v_i,
$$
where $r\leq \min\{m,n\}$, the numbers $s_i>0$ are called \emph{Schmidt coefficients}, and $\{u_i\}_{i=1}^{r}\subset\mathbb C^m$, $\{v_i\}_{i=1}^{r}\subset\mathbb C^n$ are orthonormal families. Such a representation does exist; see \cite[Section 2.5]{NIE}. The minimum possible number of terms $r$ in such a decomposition is called the \emph{Schmidt rank} of $x$ and is denoted by $\SR(x)$.
In other words, if $x=\ve(X)$ for some $X\in\mathbb M_{m,n}$, then $s_i$'s are the singular values of $X$, and $\SR(x)=\rank(X)$.

For a unit vector (bipartite pure state) $x$, let
$$
\lambda_1(x)\ge\lambda_2(x)\ge\cdots\ge\lambda_d(x)\ge0
$$
denote the squared singular values, arranged in decreasing order and counted with multiplicity, extended with zeros if necessary.
Thus
\begin{align*}
\sum_{j=1}^{d}\lambda_j(x)=1\quad\text{and}\quad\SR(x)
=\bigl|\{j:\lambda_j(x)>0\}\bigr|.
\end{align*}

For
$\mathbf A\in\mathbb M_m\otimes\mathbb M_n$, the
\emph{operator Schmidt rank} of $\mathbf A$ is
$$
\OSR(\mathbf A):=\min\left\{r:\mathbf A=\sum_{i=1}^{r}A_i\otimes B_i,\
A_i\in\mathbb M_m,\ B_i\in\mathbb M_n\right\},
$$
with the convention $\OSR(0):=0$. Clearly, $\OSR(\mathbf A^*)=\OSR(\mathbf A)$, and for each $x\in\mathbb C^m\otimes\mathbb C^n$ \begin{align}\label{eq-SR-OSR-act}
	\SR(\mathbf A x)
	\le \OSR(\mathbf A)\SR(x);
\end{align}
see \cite[Secs.~9.2 and~10.2]{BenZyc}.

In this paper, as the central result, we determine the possible Schmidt ranks of $Ux$, when $U$ has bounded operator Schmidt rank, and we obtain sharp bounds for its Schmidt spectrum. We characterize exactly which Schmidt ranks are attainable and determine the minimum operator Schmidt rank required to change the Schmidt rank from $p$ to $q$ (Theorem \ref{th-sch-reach}).

Moreover, we obtain sharp bounds on the partial sums of the squared Schmidt coefficients (Theorem \ref{th-Sch-Prof}). More precisely, we establish that if $F_r(x)$ denotes the sum of the $r$ largest squared Schmidt coefficients of $x$, then $\OSR(\mathbf U)\le k$ implies
\begin{align}\label{ine}
	F_{\min\{kr,d\}}(\mathbf Ux)\ge F_r(x)\quad\text{and}\quad F_{\min\{kr,d\}}(x)\ge F_r(\mathbf Ux).
\end{align}
We present an example of two vectors such that both have Schmidt rank $d$, while an exact conversion between them may require a unitary with large operator Schmidt rank because their Schmidt spectra are very different.
Inequalities \eqref{ine} also yield sharp bounds on Schatten-norm distortion and on the change of R\'{e}nyi entropies under unitaries of bounded operator Schmidt rank (Theorem \ref{th-Ren-OSR}), which gives a uniform sharp bound for all R\'{e}nyi parameters from the Schmidt-profile inequalities.

The notation and terminology employed in this paper are standard and largely adhere to the conventions set forth in \cite{BhMatAnal}. The Schmidt rank and operator Schmidt rank are understood in the usual tensor-rank sense; see, for example, \cite{BenZyc}.

%%%%=====================================================================

\section{Schmidt rank under operator Schmidt rank constraints}

In this section we establish our main results. We characterize precisely which Schmidt ranks are attainable and determine the minimum operator Schmidt rank needed to change the Schmidt rank from $p$ to $q$. For the method of partitioning the indices of the Schmidt coefficients and the approach of defining unitaries on orthogonal subspaces readers may consult \cite{MHN2018}.

\begin{theorem}\label{th-sch-reach}
Let $0\neq v\in\mathbb C^m\otimes\mathbb C^n$ and $p:=\SR(v)$. For any $1\le q\le d$ \begin{align}\label{eq-ex-osr-cos}
\min\left\{\OSR(\mathbf U):
\begin{array}{l}
\mathbf U\in\mathbb M_m\otimes\mathbb M_n\text{ unitary},\\
\SR(\mathbf Uv)=q\end{array}\right\}
=\left\lceil\max\left\{\frac{p}{q},\frac{q}{p}\right\}\right\rceil.
\end{align}
\end{theorem}

\begin{proof}
Set $t=\left\lceil\max\left\{\frac{p}{q},\frac{q}{p}\right\}\right\rceil$.
Let $\mathbf U$ be unitary and $\SR(\mathbf Uv)=q$. By \eqref{eq-SR-OSR-act} we have
$$
q=\SR(\mathbf Uv)\le\OSR(\mathbf U)\SR(v)
=p \OSR(\mathbf U)
$$
and
$$
p=\SR(v)=\SR(\mathbf{U}^*\mathbf Uv)\le \OSR(\mathbf U^*)\SR(\mathbf Uv)=q\OSR(\mathbf U).
$$
Hence
\begin{align}\label{eq-fU}
\OSR(\mathbf U)\ge\max\left\{\frac{p}{q},\frac{q}{p}\right\},
\end{align}
and therefore $\OSR(\mathbf U)\ge t$.

We next construct a unitary of operator Schmidt rank at most $t$. We consider two cases:\\
\noindent
\emph{Case 1: $q\le p$.}
In this case $t=\left\lceil\frac{p}{q}\right\rceil$.
Choose a Schmidt decomposition for $v$ as
\begin{align}\label{eq-sch-v}
v=\sum_{j=1}^{p}\alpha_j e_j\otimes f_j,\qquad\alpha_j>0,
\end{align}
where
$\{e_j\}_{j=1}^{p}\subseteq\mathbb C^m$ and
$\{f_j\}_{j=1}^{p}\subseteq\mathbb C^n$
are orthonormal families. Since $t=\lceil p/q\rceil$, partition $\{1,\ldots,p\}$ into $t$
nonempty, pairwise disjoint sets
$$
\{1,\ldots,p\}=\Lambda_1\sqcup\cdots\sqcup\Lambda_t
$$
such that $|\Lambda_r|\le q$ for all $1\le r\le t$, and $|\Lambda_1|=q$. Such a partition exists since $p\le tq$ and $p>(t-1)q$ whenever $t>1$.

For each $r$, write
$$
\Lambda_r=\{j_{r,1},\ldots,j_{r,s_r}\},\qquad s_r:=|\Lambda_r|\le q.
$$

Fix orthonormal vectors $g_1,\ldots,g_q\in\mathbb C^n$. For each $r$, the families
$$
f_{j_{r,1}},\ldots,f_{j_{r,s_r}}\quad\text{and}\quad g_1,\ldots,g_{s_r}
$$
are orthonormal and have the same cardinality. Hence, there exists a
unitary $V_r\in\mathbb M_n$ such that
$$
V_r f_{j_{r,a}}=g_a, \qquad 1\le a\le s_r.
$$

We now decompose the first tensor factor according to the same
partition. For $1\le r\le t$, set
$$
E_r:=\s\{e_j:j\in\Lambda_r\}\quad\text{and}\quad E_0:=\s\{e_1,\ldots,e_p\}^{\perp}.
$$
Then
$$
\mathbb C^m = (E_0\oplus E_1)\oplus E_2\oplus\cdots\oplus E_t.
$$
Let $P_1$ be the orthogonal projection onto $E_0\oplus E_1$, and,
for $2\le r\le t$, let $P_r$ be the orthogonal projection onto $E_r$.
Thus, $P_1,\ldots,P_t$ are mutually orthogonal projections satisfying
$$
\sum_{r=1}^t P_r=I_m\quad\text{and}\quad P_r e_j=e_j\quad(1\le r\le t,\,j\in\Lambda_r).
$$

Put
\begin{align}
\mathbf U:= \sum_{r=1}^{t}P_r\otimes V_r,
\end{align}
so that $\mathbf U^*\mathbf U = \sum_{r=1}^{t}P_r\otimes I_n=\mathbf I$, and similarly
$\mathbf U\mathbf U^*=\mathbf I$.
Hence, $\mathbf U$ is unitary, while $\OSR(\mathbf U)\le t$. This operator is a classic controlled unitary, which have used in quantum computing; see, for example, \cite{CohenYu2013}. From  \eqref{eq-sch-v} and the fact that the sets $\Lambda_1,\ldots,\Lambda_t$ form a
partition of $\{1,\ldots,p\}$, we derive that
\begin{align*}
\mathbf Uv
&=\left(\sum_{r=1}^{t}P_r\otimes V_r\right)
\left(\sum_{j=1}^{p}\alpha_j e_j\otimes f_j\right)\\
&=\sum_{r=1}^{t}\sum_{j\in\Lambda_r}
\alpha_j\,e_j\otimes V_r f_j\\
&=\sum_{r=1}^{t}\sum_{a=1}^{s_r}
\alpha_{j_{r,a}}\,e_{j_{r,a}}\otimes V_r f_{j_{r,a}}\\
&=\sum_{r=1}^{t}\sum_{a=1}^{s_r}\alpha_{j_{r,a}}\,
e_{j_{r,a}}\otimes g_a.
\end{align*}
Since $s_r\le q$ for every $r$, we may now collect together all
terms having the same second tensor factor $g_a$. Thus
\begin{align}\label{eq-sch-Uv}
\mathbf Uv=\sum_{a=1}^{q}\left(\sum_{\substack{1\le r\le t\\ a\le s_r}}
\alpha_{j_{r,a}}e_{j_{r,a}}\right)\otimes g_a=\sum_{a=1}^{q}x_a\otimes g_a,
\end{align}
where
$$
x_a:=\sum_{\substack{1\le r\le t\\ a\le s_r}}
\alpha_{j_{r,a}}e_{j_{r,a}},\qquad1\le a\le q.
$$
For distinct $a,b\in\{1,\ldots,q\}$, the vectors $x_a$ and $x_b$
have disjoint supports with respect to the orthonormal family
$\{e_j\}_{j=1}^{p}$. Indeed, since the sets
$\Lambda_1,\ldots,\Lambda_t$ are pairwise disjoint and the indices
$j_{r,1},\ldots,j_{r,s_r}$ are distinct within each $\Lambda_r$, no
basis vector $e_j$ can occur in both $x_a$ and $x_b$. Hence
$$
\langle x_a,x_b\rangle=0 \qquad(a\neq b).
$$
Moreover, $s_1=|\Lambda_1|=q$, so for every $1\le a\le q$ the vector
$x_a$ contains the term $\alpha_{j_{1,a}}e_{j_{1,a}}$, whose coefficient is nonzero. Therefore,  $x_a\neq0$ for every $a$. Since $g_1,\ldots,g_q$ are orthonormal and
$x_1,\ldots,x_q$ are nonzero and mutually orthogonal, the expansion in \eqref{eq-sch-Uv}
is a Schmidt decomposition up to normalization of the vectors
$x_a$. Therefore,  $\SR(\mathbf Uv)=q$.

\medskip

\noindent
\emph{Case 2: $p\le q$.}
In this case $t=\left\lceil\frac{q}{p}\right\rceil$. Write
$$
q=(t-1)p+s,\qquad 1\le s\le p.
$$
Consider the Schmidt decomposition \eqref{eq-sch-v} of $v$. Define subsets $J_1,\ldots,J_t\subseteq\{1,\ldots,p\}$ by
$$
J_t=\{1,\ldots,s\}\quad\text{and}\quad J_r=\{1,\ldots,p\},\qquad 1\le r\le t-1.
$$
When $t=1$, the first family is empty and $J_1=\{1,\ldots,p\}$, so that $s=p=q$. Also, each $J_r$ is nonempty. Thus, we always have
$$
\sum_{r=1}^{t}|J_r|=(t-1)p+s=q.
$$

For each $1\le a\le p$, define
$$
\nu_a:=\bigl|\{r\in\{1,\ldots,t\}:a\in J_r\}\bigr|.
$$
 Since $J_1=\{1,\ldots,p\}$, we have $\nu_a\ge1$ for every $1\le a\le p$. Now let
$$
\Omega:=\{(r,a):1\le r\le t,\ a\in J_r\}.
$$
For each fixed $r$, there are exactly $|J_r|$ pairs in $\Omega$
whose first coordinate is $r$. Hence, $|\Omega|=\sum_{r=1}^{t}|J_r|=q$. Therefore, $\Omega$ contains exactly $q$ ordered pairs.
Counting the same pairs according to their second coordinate gives $\sum_{a=1}^{p}\nu_a=|\Omega|=q$.

Since $\Omega$ contains exactly $q$ elements and since $q\le d=\min\{m,n\}$, we can choose $q$ orthonormal vectors in each of $\mathbb C^m$ and $\mathbb C^n$, indexed by the pairs in $\Omega$.
Thus, choose families
$$
\{u_{r,a}:(r,a)\in\Omega\}\subseteq\mathbb C^m\quad\text{and}\quad \{h_{r,a}:(r,a)\in\Omega\}
\subseteq\mathbb C^n
$$
such that
$$
\langle u_{r,a},u_{s,b}\rangle = \delta_{r,s}\delta_{a,b}
\quad\text{and}\quad \langle h_{r,a},h_{s,b}\rangle
= \delta_{r,s}\delta_{a,b}
$$
for all $(r,a),(s,b)\in\Omega$. For each $(r,a)\in\Omega$, we put
$\beta_{r,a}:=\frac{\alpha_a}{\sqrt{\nu_a}}$. This is well defined since $\nu_a\ge1$. Moreover, $\beta_{r,a}\neq0$, since every Schmidt coefficient
$\alpha_a$ in \eqref{eq-sch-v} is nonzero.

Now define
\begin{align}\label{eq-exp-vec}
w:=\sum_{(r,a)\in\Omega}\beta_{r,a}\,
u_{r,a}\otimes h_{r,a}=\sum_{r=1}^{t}\sum_{a\in J_r}
\beta_{r,a}\,u_{r,a}\otimes h_{r,a}.
\end{align}
Since both indexed families
$$
\{u_{r,a}:(r,a)\in\Omega\}\quad\text{and}\quad\{h_{r,a}:(r,a)\in\Omega\}
$$
are orthonormal, and all $q$ coefficients
$\beta_{r,a}$ are nonzero, \eqref{eq-exp-vec} is a Schmidt decomposition of $w$ with exactly $q$ nonzero terms. Thus, $\SR(w)=q$.
Moreover,
\begin{align*}
\|w\|^2=\sum_{(r,a)\in\Omega}|\beta_{r,a}|^2
=\sum_{r=1}^{t}\sum_{a\in J_r}\frac{|\alpha_a|^2}{\nu_a}
= \sum_{a=1}^{p}\nu_a\frac{|\alpha_a|^2}{\nu_a}
=\sum_{a=1}^{p}|\alpha_a|^2= \|v\|^2.
\end{align*}

For each $r\in\{1,\ldots,t\}$, the vectors $\{h_{r,a}:a\in J_r\}$ and $\{f_a:a\in J_r\}$ are orthonormal families with same cardinality. Hence, there exists a unitary $V_r\in\mathbb M_n$ such that $V_r h_{r,a}=f_a$, for any $a\in J_r$. For each $r$, let
$$
E_r:=\s\{u_{r,a}:a\in J_r\}\subseteq\mathbb C^m.
$$
Since the family $\{u_{r,a}:(r,a)\in\Omega\}$ is orthonormal, the subspaces $E_1,\ldots,E_t$ are mutually orthogonal. Therefore,  we may choose mutually orthogonal projections $P_1,\ldots,P_t\in\mathbb M_m$ such that
$$
P_r u_{r,a}=u_{r,a}\quad(a\in J_r)\quad\text{and}\quad\sum_{r=1}^{t}P_r=I_m.
$$
Define
$$
\mathbf W:=\sum_{r=1}^{t}P_r\otimes V_r
$$
so that $\mathbf W$ is unitary. Moreover, the above decomposition
contains at most $t$ elementary tensor terms, and hence $\OSR(\mathbf W)\le t$.
Employing \eqref{eq-exp-vec}, we get
\begin{align*}
\mathbf Ww&=\left(\sum_{r=1}^{t}P_r\otimes V_r\right)
\left(\sum_{(s,a)\in\Omega}\beta_{s,a}\,u_{s,a}\otimes h_{s,a}\right)\\
&=\sum_{r=1}^{t}\sum_{(s,a)\in\Omega}\beta_{s,a}\,P_r u_{s,a}\otimes V_r h_{s,a}.
\end{align*}
Since $u_{s,a}\in E_s\subseteq\RAN(P_s)$ and the projections $P_1,\ldots,P_t$ are mutually orthogonal, we have

$$
P_r u_{s,a}=\begin{cases}u_{s,a},& r=s,\\
0,& r\neq s.
\end{cases}
$$
Therefore
\begin{align*}
\mathbf Ww &=\sum_{s=1}^{t}\sum_{a\in J_s}
\beta_{s,a}\,u_{s,a}\otimes V_s h_{s,a}\\
&=\sum_{s=1}^{t}\sum_{a\in J_s}\beta_{s,a}\,u_{s,a}\otimes f_a\\
&=\sum_{a=1}^{p}\left(\sum_{\substack{1\le r\le t\\ a\in J_r}}
\beta_{r,a}u_{r,a}\right)\otimes f_a\\
&=\sum_{a=1}^{p}x_a\otimes f_a,
\end{align*}
where
$$
x_a:=\sum_{\substack{1\le r\le t\\ a\in J_r}}\beta_{r,a}u_{r,a},
\qquad1\le a\le p.
$$

In addition, $x_1,\ldots,x_p$ are pairwise orthogonal, since for $a\neq b$, we have
\begin{align}\label{eq-in-p}
\langle x_a,x_b\rangle&=\sum_{\substack{1\le r\le t\\ a\in J_r}}
\sum_{\substack{1\le s\le t\\ b\in J_s}}\overline{\beta_{r,a}}\beta_{s,b}\langle u_{r,a},u_{s,b}\rangle =0,
\end{align}
since the family $\{u_{r,a}:(r,a)\in\Omega\}$ is orthonormal. Furthermore, for each $a\in\{1,\ldots,p\}$,
\begin{align}\label{eq-nor}
\|x_a\|^2=\sum_{\substack{1\le r\le t\\ a\in J_r}}
|\beta_{r,a}|^2=\sum_{\substack{1\le r\le t\\ a\in J_r}}
\frac{|\alpha_a|^2}{\nu_a}
=\nu_a\frac{|\alpha_a|^2}{\nu_a}
=|\alpha_a|^2.
\end{align}
In particular, each $x_a$ is nonzero. Since the Schmidt coefficients in \eqref{eq-sch-v} satisfy
$\alpha_a>0$, the relations \eqref{eq-in-p} and \eqref{eq-nor} show that
$$
\left\{\frac{x_a}{\alpha_a}:1\le a\le p\right\}
$$
is an orthonormal family in $\mathbb C^m$. Since $\{e_1,\ldots,e_p\}$ is also an orthonormal family in $\mathbb C^m$, there exists a unitary $L\in\mathbb M_m$ such that $L\frac{x_a}{\alpha_a}=e_a$ for any $1\le a\le p$. We have
\begin{align*}
(L\otimes I_n)\mathbf Ww=\sum_{a=1}^{p}Lx_a\otimes f_a=\sum_{a=1}^{p}\alpha_a e_a\otimes f_a=v,
\end{align*}
where the last equality follows from the Schmidt decomposition
\eqref{eq-sch-v}.

Now consider
$$
\widetilde{\mathbf W}:=(L\otimes I_n)\mathbf W=\sum_{r=1}^{t}LP_r\otimes V_r,
$$
 which is a sum of at most $t$ elementary tensor terms, and therefore
$\OSR(\widetilde{\mathbf W})\le t$. Moreover, $\widetilde{\mathbf W}$ is unitary since both
$L\otimes I_n$ and $\mathbf W$ are unitary. Finally, we set
$\mathbf U:=\widetilde{\mathbf W}^{\,*}$ so that $\mathbf U$ is unitary and $\mathbf Uv=w$. Thus,
$$
\SR(\mathbf Uv)=\SR(w)=q.
$$
Therefore, in both cases $p\ge q$ and $p\le q$ we showed that
there exists a unitary $\mathbf U$ satisfying $\SR(\mathbf Uv)=q$ and $
\OSR(\mathbf U)\le t$, where $t=\left\lceil\max\left\{\frac{p}{q},
\frac{q}{p}\right\}\right\rceil$.
On the other hand, \eqref{eq-fU} shows that every unitary $\mathbf U$ satisfying
$\SR(\mathbf Uv)=q$ satisfies $\OSR(\mathbf U)\ge t$.
Hence, the constructed unitary necessarily satisfies $\OSR(\mathbf U)=t$, and therefore
$$
\min\left\{\OSR(\mathbf U):\mathbf U\ \text{unitary and }\
\SR(\mathbf Uv)=q\right\}=\left\lceil\max\left\{\frac{p}{q},\frac{q}{p}
\right\}\right\rceil.
$$
\end{proof}

%%---------------------------------------------------------------------------------
Theorem~\ref{th-sch-reach} determines the Schmidt ranks attainable under a prescribed bound on the operator Schmidt rank.

\begin{corollary}\label{cor-reach-ranks}
Let $0\neq v\in\mathbb C^m\otimes\mathbb C^n$ and $p:=\SR(v)$, and let $k\in\mathbb{N}$. Then
$$
\left\{\SR(\mathbf Uv):\mathbf U\ \text{unitary},\
\OSR(\mathbf U)\le k\right\}
=\left\{q\in\{1,\ldots,d\}:\left\lceil\frac{p}{k}\right\rceil\le q\le
\min\{kp,d\}\right\}.
$$
\end{corollary}

\begin{proof}
By Theorem~\ref{th-sch-reach}, a target rank $q$ is reachable with
operator Schmidt rank at most $k$ if and only if
$$
\left\lceil\max\left\{\frac{p}{q},\frac{q}{p}\right\}\right\rceil\le k.
$$
Since $k$ is an integer, this is equivalent to $p\le kq$ and $q\le kp$. For integers $1\le q\le d$, these inequalities are precisely
$$
\left\lceil\frac{p}{k}\right\rceil\le q\le \min\{kp,d\}.
$$
\end{proof}
%%%---------------------------------------------------

Corollary~\ref{cor-reach-ranks} also identifies the smallest and largest Schmidt ranks attainable under a fixed operator Schmidt rank bound.

\begin{corollary}
Let $ v\in\mathbb C^m\otimes\mathbb C^n$ be a non-zero vector with $p:=\SR(v)$ and let $k\in\mathbb{N}$. Then
\begin{align}\label{eq-min-reach}
\min_{\substack{\mathbf U\ \mathrm{unitary}\\
\OSR(\mathbf U)\le k}}\SR(\mathbf Uv)
&=\left\lceil\frac{p}{k}\right\rceil,
\end{align}
\begin{align}\label{eq-max-reach}
\max_{\substack{\mathbf U\ \mathrm{unitary}\\
\OSR(\mathbf U)\le k}}\SR(\mathbf Uv)
&=\min\{kp,d\}.
\end{align}

Thus, for every $1\le r<d$, it holds that $\SR(\mathbf Uv)\le r$ for every unitary $\mathbf U$
with $\OSR(\mathbf U)\le k$ if and only if $\SR(v)\le \left\lfloor\frac{r}{k}\right\rfloor$.
\end{corollary}

\begin{proof}
By Corollary~\ref{cor-reach-ranks}, the set of Schmidt ranks attainable
from $v$ by unitaries of operator Schmidt rank at most $k$ is
exactly
$$
\left\{q\in\{1,\ldots,d\}:\left\lceil\frac{p}{k}\right\rceil
\le q\le\min\{kp,d\}\right\}.
$$
Its smallest and largest elements are therefore $\left\lceil\frac{p}{k}\right\rceil$ and $\min\{kp,d\}$, which gives us
\eqref{eq-min-reach} and \eqref{eq-max-reach}.

Now fix $1\le r<d$. The condition $\SR(\mathbf Uv)\le r$ for every unitary $\mathbf U$ with $\OSR(\mathbf U)\le k$ is equivalent, by \eqref{eq-max-reach}, to $\min\{kp,d\}\le r$. Since $r<d$, this holds if and only if $kp\le r$, or equivalently, $p\le\left\lfloor\frac{r}{k}\right\rfloor$.
\end{proof}

%%%---------------------------------------------------------------------------

\begin{remark}
Theorem~\ref{th-sch-reach} shows that the optimal operator Schmidt rank cost of changing
Schmidt rank from $p$ to $q$ depends only on the integers $p$ and
$q$, and not on the Schmidt coefficients of the initial vector. Thus,
the rank-reachability problem depends only on discrete rank information.

A bound on operator Schmidt rank, however, does not capture every constraint that arises. We therefore shift our attention to the Schmidt coefficients.

\end{remark}

%%%-------------------------------------------------------------------------

Let $x\in\mathbb C^m\otimes\mathbb C^n$ be a unit vector, $d=\min\{m,n\}$, and $\lambda_1(x)\ge\cdots\ge\lambda_d(x)\ge0$ denote its squared Schmidt coefficients. For $1\le r\le d$, set
$$
F_r(x):=\sum_{j=1}^r\lambda_j(x).
$$
We refer to $(F_1(x),\dots,F_d(x))$ as the \emph{$r$th cumulative Schmidt profile}. Thus, $F_r(x)^{1/2}$ is the vector $s(r)$-norm
\cite{Johkri2010}. In particular, $F_r(x)=1$ if and only if $\SR(x)\le r$.
%%%-------------------------------------------------------------------------

\medskip
The next theorem shows how a bound on operator Schmidt rank affects the partial sums of the squared Schmidt coefficients.
\begin{theorem}\label{th-Sch-Prof}
Let $\mathbf U\in\mathbb M_m\otimes\mathbb M_n$ be unitary with
$\OSR(\mathbf U)\le k$ for $k\in\mathbb{N}$. Then, for every unit vector
$x\in\mathbb C^m\otimes\mathbb C^n$ and every $1\le r\le d$,
$$
F_s(\mathbf Ux)\ge F_r(x)\quad\text{and}\quad F_s(x)\ge F_r(\mathbf Ux),
$$
where $s:=\min\{kr,d\}$. Moreover, the index $s$ is optimal, in the sense that whenever $s>1$, there exist such
$x$ and $\mathbf U$ for which
$$
F_{s-1}(\mathbf Ux)<F_r(x).
$$
\end{theorem}

\begin{proof}
Set $s:=\min\{kr,d\}$, and for $1\le j\le d$ let
$$
\Sigma_j:=\{z\in\mathbb C^m\otimes\mathbb C^n:\SR(z)\le j\}.
$$

For a  unit vector $y$,  with its coefficient matrix
$Y$, elements of $\Sigma_j$ correspond exactly to
matrices of rank at most $j$. The singular values of $Y$
are the Schmidt coefficients of $y$, so that $\sigma_i(Y)^2=\lambda_i(y)$. The Eckart--Young--Mirsky theorem for the truncated singular value decomposition (see, for example, \cite{GolVanLoa}) ensures that the
best approximation of $Y$ by a matrix of rank at most $j$ is obtained
by keeping the first $j$ singular values. In other words, the best
approximation of $y$ by a vector of Schmidt rank at most $j$ is
obtained by truncating its Schmidt decomposition. Hence,
$$
\DS(y,\Sigma_j)^2=\sum_{i=j+1}^{d}\sigma_i(Y)^2
=\sum_{i>j}\lambda_i(y)=1-F_j(y).
$$
Hence, by \eqref{eq-SR-OSR-act}, $\mathbf U\Sigma_r\subseteq\Sigma_s$. Therefore
\begin{align*}
1-F_s(\mathbf Ux)
&=\DS(\mathbf Ux,\Sigma_s)^2\\
&\le \DS(\mathbf Ux,\mathbf U\Sigma_r)^2 \\
&=\DS(x,\Sigma_r)^2\qquad(\mbox{since $\mathbf U$ is unitary})\\
&=1-F_r(x).
\end{align*}
Thus, $F_s(\mathbf Ux)\ge F_r(x)$. Since
$\OSR(\mathbf U^*)=\OSR(\mathbf U)\le k$,
applying the same argument to $\mathbf U^*$ gives $F_s(x)\ge F_r(\mathbf Ux)$.

To prove optimality, if $s=1$, there is nothing to show.
Suppose $s\ge2$, and choose a unit vector $x$ with
$\SR(x)=r$. By Theorem~\ref{th-sch-reach}, there exists
a unitary $\mathbf U$ such that
$$
\SR(\mathbf Ux)=s,\qquad
\OSR(\mathbf U)=\left\lceil\frac{s}{r}\right\rceil\le k.
$$
Therefore,  $F_r(x)=1$ and $F_{s-1}(\mathbf Ux)<1$ so that $F_{s-1}(\mathbf Ux)<F_r(x)$. Thus, $s$ cannot be replaced by a smaller universal index in the first
inequality. Applying the same construction to the pair
$(\mathbf Ux,\mathbf U^*)$ shows optimality of the second inequality.
\end{proof}
%%%--------------------------------------------------------------

Theorem~\ref{th-Sch-Prof} reveals that Schmidt rank does not fully determine the constraints imposed by low-OSR unitaries. In particular, the next example shows that two states with the same Schmidt rank can still behave differently under low-OSR unitary transformations

\begin{example}
Assume $d\ge 2$, and choose orthonormal families $e_1,\ldots,e_d\in\mathbb C^m$ and $f_1,\ldots,f_d\in\mathbb C^n$. For $0<\varepsilon<1/d$, set
$$
x_\varepsilon=\sqrt{1-(d-1)\varepsilon}\,e_1\otimes f_1
+\sqrt{\varepsilon}\sum_{j=2}^d e_j\otimes f_j\quad\text{and}\quad y
=\frac{1}{\sqrt d}\sum_{j=1}^d e_j\otimes f_j.
$$
Both vectors have Schmidt rank $d$, so Schmidt rank alone gives no
nontrivial obstruction to a unitary conversion $x_\varepsilon\mapsto y$.

Suppose that $\mathbf Ux_\varepsilon=y$, and set $k:=\OSR(\mathbf U)$. Since
$$
F_1(x_\varepsilon)=1-(d-1)\varepsilon \quad\text{and}\quad
F_s(y)=\frac{s}{d},
$$
applying Theorem~\ref{th-Sch-Prof} with $r=1$ gives
$$
\frac{\min\{k,d\}}{d}=F_{\min\{k,d\}}(y)
\ge F_1(x_\varepsilon)=1-(d-1)\varepsilon.
$$
Thus, $\OSR(\mathbf U)\ge \left\lceil d\bigl(1-(d-1)\varepsilon\bigr)\right\rceil$. In particular, if
$0<\varepsilon<\frac{1}{d(d-1)}$, then $\OSR(\mathbf U)\ge d$.

Thus, even when two vectors have the same maximal Schmidt rank, an exact conversion between the two specific vectors may require a unitary with large operator Schmidt rank. The obstruction comes from the Schmidt spectrum itself, not merely from the number of nonzero Schmidt coefficients.
\end{example}

%%%%------------------------------------------------------------------------------
\section{Applications}
In this section, we present several applications. We first utilize the Schmidt-profile bounds to obtain sharp Schatten norms bounds.

 For a linear map
$\Phi:\mathbb M_{m,n}\to\mathbb M_{m,n}$, we use the notation
$$
\|\Phi\|_{(p,p)}:=\sup_{X\ne0}\frac{\|\Phi(X)\|_p}{\|X\|_p}.
$$
For $\mathbf U\in\mathbb M_m\otimes\mathbb M_n$, let
$\Phi_{\mathbf U}:\mathbb M_{m,n}\to\mathbb M_{m,n}$ be the linear map
determined by
\begin{align}\label{eq-phiu}
\Phi_{\mathbf U}(X):=\ve^{-1}\!\bigl(\mathbf U\,\ve(X)\bigr).
\end{align}
Thus, $\ve\circ\Phi_{\mathbf U}=\mathbf U\circ\ve$.

Theorem~\ref{th-Sch-Prof} gives the following sharp Schatten-norm bounds for $\Phi_{\mathbf U}$.

\begin{theorem}\label{th-Sch-Norm}
Let $\mathbf U\in\mathbb M_m\otimes\mathbb M_n$ be unitary and set
$$
\ell:=\min\{\OSR(\mathbf U),d\}.
$$
Then, for every $1\le p\le\infty$ and $X\in\mathbb M_{m,n}$,
\begin{align}\label{eq-ced}
\ell^{-\left|\frac1p-\frac12\right|}\|X\|_p
\le \|\Phi_{\mathbf U}(X)\|_p
\le \ell^{\left|\frac1p-\frac12\right|}\|X\|_p.
\end{align}
Moreover, the bound is sharp: for every integer $k\ge1$,
\begin{align}\label{eq-ot}
\sup_{\substack{\mathbf U\ {\rm unitary}\\
\OSR(\mathbf U)\le k}}\|\Phi_{\mathbf U}\|_{(p,p)}
=\min\{k,d\}^{\left|\frac1p-\frac12\right|}.
\end{align}
\end{theorem}
\begin{proof}
Let $X\neq0$ and $Y:=\Phi_{\mathbf U}(X)$. Put $c=\|X\|_2$.
Since vectorization is an isometry, it follows from \eqref{eq-phiu} that
$\|Y\|_2=\|X\|_2=c$. Define
$$
x:=\frac{\ve(X)}{c}
\quad\text{and}\quad y:=\frac{\ve(Y)}{c}
=\mathbf Ux.
$$
The squared Schmidt coefficients of $x$ and $y$ are
$s_j(X)^2/c^2$ and $s_j(Y)^2/c^2$, respectively. Applying
Theorem~\ref{th-Sch-Prof} with $r=1$ gives
\begin{align}\label{eq-f1}
F_1(y)\le F_\ell(x).
\end{align}
 Moreover, since the squared Schmidt coefficients are nonincreasing,
$$
\lambda_j(x)\le \lambda_1(x)=F_1(x), \qquad 1\le j\le \ell.
$$
Therefore,
\begin{align}\label{eq-fel}
F_\ell(x) =\sum_{j=1}^{\ell}\lambda_j(x)
\le \ell\,\lambda_1(x)=\ell F_1(x).
\end{align}
From \eqref{eq-f1} and \eqref{eq-fel} we get $\lambda_1(y)\le\ell\,\lambda_1(x)$. Returning to singular values gives
$$
\frac{\|Y\|_\infty^2}{c^2}\le\ell\frac{\|X\|_\infty^2}{c^2}
\qquad\text{or}\qquad\|\Phi_{\mathbf U}(X)\|_\infty
\le \ell^{1/2}\|X\|_\infty,
$$
whence
$$
\|\Phi_{\mathbf U}\|_{(\infty,\infty)}\le \ell^{1/2}.
$$

With respect to the Hilbert--Schmidt inner product, the map
$\ve$ is unitary and
\begin{align}\label{eq-phiu-adj}
\Phi_{\mathbf U}^*=\Phi_{\mathbf U^*}.
\end{align}
Indeed, this follows from $\Phi_{\mathbf U}=\ve^{-1}\mathbf U\ve$.
Since
$\OSR(\mathbf U^*)=\OSR(\mathbf U)$,
Schatten-norm duality yields
$$
\|\Phi_{\mathbf U}\|_{(1,1)}=\|\Phi_{\mathbf U^*}\|_{(\infty,\infty)}
\le \ell^{1/2}.
$$
Moreover, from $\|X\|_2=\|Y\|_2$ we have $\|\Phi_{\mathbf U}\|_{(2,2)}=1$.

By the Riesz--Thorin interpolation theorem for Schatten norms
(see, for example, \cite[Theorems~2.9--2.10 and Remark~1, p.~23]{Simon2005}), if a linear map
$\Phi$ satisfies
$$
\|\Phi(X)\|_{p_0}\le M_0\|X\|_{p_0} \quad\text{and}\quad
\|\Phi(X)\|_{p_1}\le M_1\|X\|_{p_1},
$$
then, for
$$
\frac1p=\frac{1-\theta}{p_0}+\frac{\theta}{p_1},\qquad 0\le\theta\le1,
$$
we have
$$
\|\Phi(X)\|_p\le M_0^{1-\theta}M_1^\theta\|X\|_p.
$$

Applying this to the endpoint estimates above, for $1\le p\le2$
we interpolate between $p=1$ and $p=2$ to obtain
$$
\|\Phi_{\mathbf U}(X)\|_p \le\ell^{\frac1p-\frac12}\|X\|_p.
$$
For $2\le p\le\infty$, interpolation between $p=2$ and
$p=\infty$ gives
$$
\|\Phi_{\mathbf U}(X)\|_p \le\ell^{\frac12-\frac1p}\|X\|_p.
$$
Thus
$$
\|\Phi_{\mathbf U}(X)\|_p\le\ell^{\left|\frac1p-\frac12\right|}\|X\|_p,
\qquad 1\le p\le\infty.
$$

Since $\Phi_{\mathbf U}^{-1}=\Phi_{\mathbf U^*}$ and $\mathbf U^*$ has the same operator Schmidt rank as $\mathbf U$, applying the preceding estimate to $\Phi_{\mathbf U^*}$ gives
$$
\|X\|_p\le\ell^{\left|\frac1p-\frac12\right|}\|\Phi_{\mathbf U}(X)\|_p,
$$
which yields \eqref{eq-ced}.

We now prove sharpness. Fix $k\ge1$ and set $t:=\min\{k,d\}$. Choose orthonormal families $e_1,\ldots,e_t\in\mathbb C^m$ and $f_1,\ldots,f_t\in\mathbb C^n$ and let
$$
a:=\frac1{\sqrt t}\sum_{j=1}^t e_j.
$$
Set
$$
P_1:=I_m-\sum_{j=2}^t e_je_j^*\quad\text{and}\quad P_j:=e_je_j^* \quad (2\le j\le t),
$$
and choose unitaries $V_j\in\mathbb M_n$ satisfying
$V_jf_1=f_j$ for $1\le j\le t$. Then $\mathbf V:=\sum_{j=1}^t P_j\otimes V_j$ is unitary, since the projections $P_1,\ldots,P_t$ are mutually
orthogonal and sum to $I_m$. Moreover, $\OSR(\mathbf V)\le t$, and
$$
\mathbf V(a\otimes f_1)=\frac1{\sqrt t}\sum_{j=1}^t e_j\otimes f_j.
$$
Let $X_0,Y_0\in\mathbb M_{m,n}$ be determined by
$$
\ve(X_0)=a\otimes f_1
\quad\text{and}\quad\ve(Y_0)
=\frac1{\sqrt t}\sum_{j=1}^t e_j\otimes f_j.
$$
Since $\ve(X_0)$ is a unit product vector,
it has a single nonzero singular value equal to $1$, and hence $\|X_0\|_p=1$. Moreover,
$\ve(Y_0)$ is already a Schmidt decomposition, so $Y_0$ has exactly $t$ nonzero
singular values, all equal to $t^{-1/2}$. Therefore,
$$
\|Y_0\|_p=\left(t\cdot t^{-p/2}\right)^{1/p}
=t^{\frac1p-\frac12}.
$$
Also, we have $Y_0=\Phi_{\mathbf V}(X_0)$. Thus, for
$1\le p\le2$ we have $\|\Phi_{\mathbf V}\|_{(p,p)}
\ge t^{\frac1p-\frac12}$, whereas for $2\le p\le\infty$ we have $\|\Phi_{\mathbf V^*}\|_{(p,p)}
\ge t^{\frac12-\frac1p}$.
Since both $\mathbf V$ and $\mathbf V^*$ have operator Schmidt rank
at most $t\le k$, we obtain
$$
\sup_{\substack{\mathbf U\ {\rm unitary}\\
\OSR(\mathbf U)\le k}}\|\Phi_{\mathbf U}\|_{(p,p)}
\ge t^{\left|\frac1p-\frac12\right|}.
$$
This together with \eqref{eq-ced} gives \eqref{eq-ot}.
\end{proof}

%%-----------------------------------------------------------------------------------------------

The Schatten-norm estimate uses only the case $r=1$ of Theorem~\ref{th-Sch-Prof}. The full family of inequalities also gives information about entropic quantities.

 For $0<\alpha<\infty$, $\alpha\neq1$, the \emph{R\'enyi entropy} of a unit vector $x\in\mathbb C^m\otimes\mathbb C^n$ is defined by
$$
H_\alpha(x):=\frac{1}{1-\alpha}\log\left(\sum_{j=1}^d\lambda_j(x)^\alpha\right),
$$
with the usual continuous extensions
$$
H_0(x)=\log\SR(x),\quad H_1(x)=-\sum_{j=1}^d\lambda_j(x)\log\lambda_j(x),\quad \mbox{and} \quad
H_\infty(x)=-\log\lambda_1(x);
$$
see \cite[Chapter 7]{BKMS} for more information on quantum entropies and quantum relative entropies as well as \cite{SEO1, SEO2}. Thus, $H_1$ is the usual entropy of entanglement of a bipartite pure
state, while $H_0$ and $H_\infty$ measure, respectively, the size
of its Schmidt support and its largest Schmidt weight.

 For $\alpha=0$, the bound follows from the multiplicative control of Schmidt rank. For $\alpha=1$, related entanglement-generation quantities for bipartite unitaries were investigated in \cite{BenHarLeSm2003,Nielsen2003}. In what follows, we show that for every $\alpha\in[0,\infty]$,
$$
\bigl|H_\alpha(\mathbf Ux)-H_\alpha(x)\bigr| \le \log \ell, \qquad \ell=\min\{\OSR(\mathbf U),d\}.
$$

For a probability vector $b=(b_1,\ldots,b_d),\,\, b_1\ge\cdots\ge b_d\ge0$,
and an integer $1 \leq \ell \leq d$, define its $\ell$-block aggregating by
$\mathcal B_\ell(b):=(c_1,\ldots,c_d)$, where
$$
c_r :=
\begin{cases}
	\displaystyle \sum_{j=(r-1)\ell+1}^{\min\{r\ell,d\}} b_j, &\quad 1 \le r \le \left\lceil \frac{d}{\ell} \right\rceil, \\
	0, &\quad \left\lceil \frac{d}{\ell} \right\rceil < r \leq d .
\end{cases}
$$
Thus, $\mathcal B_\ell(b)$ is obtained by summing consecutive blocks of
$\ell$ entries of $b$, with the final block truncated if necessary. This is a
particular contiguous aggregation in the sense of
\cite{CicVac2018}.
Since $b_1\ge b_2\ge\cdots\ge b_d$, the block sums also form a
nonincreasing sequence. Hence, $\mathcal B_\ell(b)$ is again a probability vector written in
decreasing order.

%%%%-------------------------------------------------------------------------
Recall that  for probability vectors
$a,b\in\mathbb R^d$ written in decreasing order,  $a\succ b$
means
$$
\sum_{j=1}^r a_j\ge \sum_{j=1}^r b_j,
\qquad 1\le r<d.
$$
\begin{lemma}\label{lm-comp-maj}
Let $x,y$ be unit vectors and let $1\le \ell\le d$. Suppose that
$$
F_{\min\{\ell r,d\}}(y)\ge F_r(x),\qquad
1\le r\le \left\lceil\frac d\ell\right\rceil.
$$
Then $\mathcal B_\ell(\lambda(y))\succ \lambda(x)$.
\end{lemma}

\begin{proof}
Set $c:=\mathcal B_\ell(\lambda(y))$. By the definition of the block aggregating, each entry $c_j$ collects
$\ell$ consecutive components of $\lambda(y)$, except possibly the
last block, which may contain fewer than $\ell$ components. Therefore,
for every $r\ge1$,
$$
\sum_{j=1}^{r}c_j=\sum_{i=1}^{\min\{\ell r,d\}}\lambda_i(y).
$$

Note that the first $r$ blocks of
$\mathcal B_\ell(\lambda(y))$ contain precisely the first
$\min\{\ell r,d\}$ entries of $\lambda(y)$. Hence,
$$
\sum_{j=1}^r c_j= \sum_{j=1}^r \bigl(\mathcal B_\ell(\lambda(y))\bigr)_j
=\sum_{j=1}^{\min\{\ell r,d\}} \lambda_j(y)=
F_{\min\{\ell r,d\}}(y)\ge F_r(x)=\sum_{j=1}^r\lambda_j(x).
$$
For $r>\lceil d/\ell\rceil$, all nonzero entries of $c$ have already
been included, so
$$
\sum_{j=1}^r c_j=1\ge \sum_{j=1}^r\lambda_j(x).
$$
For $r=d$ the two total sums are both equal to $1$. Hence
$c\succ\lambda(x)$, as required.
\end{proof}
%%---------------------------------------------------------------

The majorization relation above allows us to translate the spectral
constraints into bounds on entropic quantities. We use it to obtain sharp
bounds on the change of R\'{e}nyi entropies under unitaries with bounded
operator Schmidt rank.
\begin{theorem}\label{th-Ren-OSR}
Let $\mathbf U$ be a bipartite unitary on
$\mathbb C^m\otimes\mathbb C^n$, and let
$$
\ell:=\min\{\OSR(\mathbf U),d\}.
$$
Then, for every unit vector
$x\in\mathbb C^m\otimes\mathbb C^n$ and every
$\alpha\in[0,\infty]$,
$$
\bigl|H_\alpha(\mathbf Ux)-H_\alpha(x)\bigr|\le\log \ell.
$$
Moreover, this bound is sharp. For every integer
$1\le k\le d$,
$$
\sup_{\substack{\mathbf U\ \mathrm{unitary}\\
\OSR(\mathbf U)\le k}}\ \sup_{\|x\|=1}
\bigl|H_\alpha(\mathbf Ux)-H_\alpha(x)\bigr|
=\log k,
$$
for every $\alpha\in[0,\infty]$.
\end{theorem}

\begin{proof}
 Set $\ell:=\min\{\OSR(\mathbf U),d\}$.
Since $\OSR(\mathbf U)\ge\ell$ and every Schmidt rank is at most $d$,
Theorem~\ref{th-Sch-Prof} gives, for every $1\le r\le d$,
\begin{equation*}
	F_{\min\{\ell r,d\}}(\mathbf Ux)
	\ge F_r(x),
\end{equation*}
where we used $\min\{\ell r,d\}=\min\{\OSR(\mathbf U)r,d\}$. Applying Lemma~\ref{lm-comp-maj} with $y=\mathbf Ux$ we obtain
$$
\mathcal B_\ell(\lambda(\mathbf Ux))\succ\lambda(x).
$$
Set $b:=\lambda(\mathbf Ux)$ and $c:=\mathcal B_\ell(b)$ and let $B_1,\ldots,B_N$ denote the consecutive blocks used in the definition of the block aggregation. By definition, we have
$$
c_r=\sum_{i\in B_r}b_i,\qquad r=1,\ldots,N,
$$
and each block satisfies $|B_r|\le \ell$. Since $0<\alpha<1$, the function $t^\alpha$ is concave and the Jensen's inequality gives for each block $B_r$ that
 $$
\sum_{i\in B_r}b_i^\alpha\le|B_r|^{1-\alpha}
\left(\sum_{i\in B_r}b_i\right)^\alpha \le \ell^{1-\alpha}c_r^\alpha.
$$
Summing over all blocks gives
\begin{align}\label{eq-Bux}
\sum_{j=1}^d b_j^\alpha
=\sum_{r=1}^N\sum_{i\in B_r}b_i^\alpha
\le\ell^{1-\alpha}\sum_{r=1}^N c_r^\alpha
=\ell^{1-\alpha}\sum_{j=1}^d c_j^\alpha.
\end{align}
On the other hand, since $c=\mathcal B_\ell(b)\succ\lambda(x)$, the majorization inequality for concave functions (see \cite{BhMatAnal}) gives
\begin{align}\label{eq-Bux2}
\sum_{j=1}^d c_j^\alpha
\le\sum_{j=1}^d\lambda_j(x)^\alpha.
\end{align}
Combining \eqref{eq-Bux} and \eqref{eq-Bux2} we obtain
$$
\sum_{j=1}^d b_j^\alpha\le
\ell^{1-\alpha} \sum_{j=1}^d\lambda_j(x)^\alpha.
$$

Taking logarithms and dividing by $1-\alpha>0$ yields
\begin{align}\label{eq-ent1}
H_\alpha(\mathbf Ux)\le H_\alpha(x)+\log\ell.
\end{align}
Now suppose that $\alpha>1$. Since $t^\alpha$ is convex, Jensen's
inequality gives, for each block $B_r$,
$$
\sum_{i\in B_r}b_i^\alpha\ge|B_r|^{1-\alpha}
\left(\sum_{i\in B_r}b_i\right)^\alpha \ge \ell^{1-\alpha}c_r^\alpha,
$$
where the right inequality follows from $|B_r|\le\ell$ and $1-\alpha<0$.
Summing over all blocks yields
\begin{align}\label{eq-Bix}
\sum_{j=1}^d b_j^\alpha
=\sum_{r=1}^N\sum_{i\in B_r}b_i^\alpha
\ge\ell^{1-\alpha}\sum_{r=1}^N c_r^\alpha
=\ell^{1-\alpha}\sum_{j=1}^d c_j^\alpha.
\end{align}
Moreover, since $c=\mathcal B_\ell(b)\succ\lambda(x)$ and $t^\alpha$ is convex, the majorization inequality gives
\begin{align}\label{eq-Bix2}
\sum_{j=1}^d c_j^\alpha\ge\sum_{j=1}^d\lambda_j(x)^\alpha .
\end{align}
Combining \eqref{eq-Bix} and \eqref{eq-Bix2} we obtain
$$
\sum_{j=1}^d b_j^\alpha
\ge \ell^{1-\alpha}\sum_{j=1}^d\lambda_j(x)^\alpha.
$$
Since $1-\alpha<0$, taking logarithms and dividing by
$1-\alpha$ reverses the inequality. Hence,
$$
H_\alpha(\mathbf Ux)\le H_\alpha(x)+\log\ell.
$$
Applying the same argument to $\mathbf U^*$, and using $\OSR(\mathbf U^*)
=\OSR(\mathbf U)$ gives
\begin{align}\label{eq-ent2}
H_\alpha(x)\le H_\alpha(\mathbf Ux)+\log\ell.
\end{align}
Combining \eqref{eq-ent1} and \eqref{eq-ent2} yields
$$
\bigl|H_\alpha(\mathbf Ux)-H_\alpha(x)\bigr|
\le\log\ell,\qquad0<\alpha<\infty,\ \alpha\neq1.
$$
For $\alpha=0$, \eqref{eq-SR-OSR-act} gives
\begin{equation*}
	\SR(\mathbf Ux)
	\le \OSR(\mathbf U)\SR(x)
\end{equation*}
and, by applying the same inequality to $\mathbf U^*$,
\begin{equation*}
	\SR(x)
	\le \OSR(\mathbf U)\SR(\mathbf Ux).
\end{equation*}
Since $\SR(x),\SR(\mathbf Ux)\le d$, these inequalities imply
\begin{equation*}
	\frac1\ell
	\le
	\frac{\SR(\mathbf Ux)}{\SR(x)}
	\le
	\ell.
\end{equation*}
Consequently,
\begin{equation*}
	\bigl|H_0(\mathbf Ux)-H_0(x)\bigr|
	=
	\left|
	\log\frac{\SR(\mathbf Ux)}{\SR(x)}
	\right|
	\le\log\ell.
\end{equation*}

For $\alpha=1$, the asserted inequality follows by letting
$\alpha\to1$ in the already established inequality for
$0<\alpha<\infty$, $\alpha\ne1$, because the Rényi entropies converge
to the Shannon entropy in finite dimension. Finally, for
$\alpha=\infty$, the profile inequality with $r=1$ gives
$$
\lambda_1(\mathbf Ux)=F_1(\mathbf Ux)
 \le F_\ell(x)\le \ell F_1(x)=\ell\lambda_1(x),
$$
and the same inequality applied to $\mathbf U^*$ gives
$$
\lambda_1(x)\le \ell\lambda_1(\mathbf Ux).
$$
Taking negative logarithms yields
$$
\bigl|H_\infty(\mathbf Ux)-H_\infty(x)\bigr|\le\log\ell.
$$
Thus, the same bound holds for all $\alpha\in[0,\infty]$.

For sharpness, fix $1\le k\le d$. Use the unitary $\mathbf V$ constructed
in the proof of Theorem~\ref{th-Sch-Norm}, with
$\OSR(\mathbf V)\le k$, and let $x_0$ be the product vector satisfying
\begin{equation*}
	\mathbf Vx_0
	=
	\frac1{\sqrt k}\sum_{j=1}^{k}e_j\otimes f_j.
\end{equation*}
The squared Schmidt coefficients of $x_0$ are $(1,0,\ldots,0)$, whereas those of $\mathbf Vx_0$ are
$\left(\frac1k,\ldots,\frac1k,0,\ldots,0\right)$, with exactly $k$ nonzero entries. Hence, for every
$\alpha\in[0,\infty]$,
\begin{equation*}
	H_\alpha(x_0)=0
	\quad \mbox{and} \quad
	H_\alpha(\mathbf Vx_0)=\log k.
\end{equation*}
Thus
\begin{equation*}
	\sup_{\substack{\mathbf U\ \mathrm{unitary}\\
			\OSR(\mathbf U)\le k}}
	\ \sup_{\|x\|=1}
	\bigl|H_\alpha(\mathbf Ux)-H_\alpha(x)\bigr|
	\ge\log k.
\end{equation*}
The reverse inequality follows from the already established estimate with
$\ell\le k$, and therefore equality holds.
\end{proof}
%%%%--------------------------------------------------------------------------
\begin{remark}
Let $\mathcal K_A$ and $\mathcal K_B$ be finite-dimensional ancillary
Hilbert spaces, and let $\widetilde{\mathbf U}$ denote the unitary on
$$
(\mathbb C^m\otimes\mathcal K_A)\otimes
(\mathbb C^n\otimes\mathcal K_B)
$$
obtained by letting $\mathbf U$ act on
$\mathbb C^m\otimes\mathbb C^n$ and as the identity on
$\mathcal K_A\otimes\mathcal K_B$, under the canonical reordering of
tensor factors. Then $\OSR(\widetilde{\mathbf U})=\OSR(\mathbf U)$.
Indeed, if $\mathbf U=\sum_{j=1}^r A_j\otimes B_j$ is a minimal operator Schmidt decomposition, then $\widetilde{\mathbf U}=\sum_{j=1}^r
(A_j\otimes I_{\mathcal K_A})\otimes(B_j\otimes I_{\mathcal K_B})$ and tensoring with a fixed nonzero identity preserves linear independence of the two factor families.

Consequently, Theorem~\ref{th-Ren-OSR} applies on the enlarged
bipartite space. More precisely, if
$$
d':=\min\bigl\{m\dim\mathcal K_A,\,n\dim\mathcal K_B\bigr\},
$$
then, for every unit vector $x\in
(\mathbb C^m\otimes\mathcal K_A)\otimes(\mathbb C^n\otimes\mathcal K_B)$ and every $\alpha\in[0,\infty]$, we have
$$
\bigl|H_\alpha(\widetilde{\mathbf U}x)-H_\alpha(x)\bigr|\le
\log\min\{\OSR(\mathbf U),d'\},
$$
where the R\'enyi entropies are computed with respect to the enlarged
bipartition.
\end{remark}

%==========================================================================================%
%==========================================================================================%
%==========================================================================================%

\medskip
\noindent \textit{Author Contributions Statement.} All authors wrote, edited, and reviewed the manuscript.

\medskip
\noindent \textit{Conflict of Interest Statement.} On behalf of the authors, the corresponding author states that there is no conflict of interest.

\medskip
\noindent\textit{Data Availability Statement.} Data sharing is not applicable to this article as no datasets were generated or analyzed during the current study.

\medskip
\noindent \textit{Funding Declaration.} This research received no funding.

\medskip

%==========================================================================================%
%==========================================================================================%
%==========================================================================================%

\end{document}